\documentclass{article}

\usepackage[margin=1in]{geometry}
\usepackage{amsmath,amssymb,amsfonts,amsthm}

\newtheorem{theorem}{Theorem} 

\newtheorem{definition}{Definition}

\usepackage{graphicx}        
\usepackage{epstopdf}
\usepackage{tikz,xcolor}

\usepackage{multicol}        

\usepackage[stable]{footmisc}
\usepackage{url}

\newcommand{\bc}{\begin{center}}
\newcommand{\ec}{\end{center}}
\newcommand{\bdefn}{\bigskip \begin{definition}}
\newcommand{\edefn}{\end{definition} \bigskip}
\newcommand{\benum}{\begin{enumerate}}
\newcommand{\eenum}{\end{enumerate}}
\newcommand{\beq}{\begin{equation}}
\newcommand{\eeq}{\end{equation}}
\newcommand{\bfig}{\begin{figure}[htb] \centering}
\newcommand{\efig}{\end{figure}}
\newcommand{\bitem}{\begin{itemize}}
\newcommand{\eitem}{\end{itemize}}
\newcommand{\bray}{\begin{array}}
\newcommand{\eray}{\end{array}}
\newcommand{\btab}{\begin{table}[hbt]}
\newcommand{\etab}{\end{table}}
\newcommand{\bthm}{\bigskip \begin{theorem}}
\newcommand{\ethm}{\end{theorem} \bigskip}

\newcommand{\tr}{\operatorname{tr}}

\newcommand{\rr}{{\cal{R}}_0}

\newcommand{\rv}{{\cal{R}}_v}
\newcommand{\ihat}{\widehat{I}}

\DeclareSymbolFont{AMSb}{U}{msb}{m}{n}
\DeclareMathSymbol{\R}{\mathbin}{AMSb}{"52}

\definecolor{fgreen}{RGB}{0, 128, 0}
\definecolor{borange}{RGB}{255, 80, 0}
\definecolor{dbrown}{RGB}{140, 60, 30}

\newcommand{\red}[1]{{#1}}

\title{The Impact of Neglecting Vaccine Unwillingness\\ in Epidemiology Models}
\author{Glenn Ledder}
\date{}

\begin{document}

\maketitle

\begin{abstract}
\red{With significant population fractions in many societies who refuse vaccines, it is important to reconsider how vaccination is incorporated into compartmental epidemiology models.  It is still most common to apply the vaccination rate to the entire class of susceptibles, rather than to use the more realistic assumption that the vaccination rate function should depend only on the population of susceptibles who are willing and able to receive a vaccination.  This study uses a simple generic disease model to address two questions: (1) How much error is introduced in key model outcomes by neglecting vaccine unwillingness?, and (2) Can the error be reduced by incorporating vaccine unwillingness into the vaccination rate constant rather than the rate diagram?  The answers depend greatly on the time scale of interest.  For the endemic time scale, where longterm behavior is studied with equilibrium point analysis, the error in neglecting unwillingess is large and cannot be improved upon by decreasing the vaccination rate constant.  For the epidemic time scale, where the first big epidemic wave is studied with numerical simulations, the error can still be significant, particularly for diseases that are relatively less infectious and vaccination programs that are relatively slow.}
\end{abstract}

\section{Introduction}

Vaccination is one of many processes that can be incorporated into epidemiology models.  Its effect on the rate diagram varies because the benefits of vaccination depends on the disease and the vaccine.  In the best scenario, vaccination confers guaranteed and permanent immunity.  Lesser benefits are also common:
\benum
\item
Temporary or partial immunity;
\item
Decreased probability of being infected;
\item
Decreased severity of infection;
\item
Faster recovery;
\item
Decreased risk of infecting others.
\eenum
Historically, it has been standard to incorporate vaccination as a linear process with rate proportional to the full susceptible population.  In 2022, Ledder \cite{ledder2022mbe} presented the idea that vaccination should only be applied to a subclass of willing susceptibles, taking into account that some susceptibles are either unwilling or unable to be vaccinated.  \red{Numerous studies, both before and after the onset of the COVID-19 pandemic, have shown that the fraction of individuals who are unwilling or unable is higher than might previously have been thought based on public health policies, as there has always been a significant minority of people who do not follow those policies \cite{dube2013hesitancy}.  A study published by the Annenberg Foundation in 2023 showed that only 71\% of Americans thought vaccines were safe (and this presumably included some who thought vaccines are not effective) and that only 63\% of Americans at that time thought that the COVID vaccine was less dangerous than the COVID disease \cite{annenberg2023vaccine}.  Given recent trends \cite{cdc2025immunization,pbs2025rfkjr,stolberg2025kennedy,hhs2025mrna}, it is reasonable to guess that 30\% is a conservative estimate of the fraction of Americans who will be unwilling to receive a vaccine for the next epidemic.  While other societies may exhibit different attitudes toward vaccination, it is reasonable to expect that some level of vaccine unwillingness is present everywhere.}

Restricting vaccination to a subclass of susceptibles makes a model more complicated.  At minimum, we need an additional state variable, either for willing susceptibles or unwilling susceptibles, and, depending on the biological assumptions about the disease, classes other than the susceptibles may also need to be subdivided.  This raises the question of whether the extra complexity is justified by the improved realism.  This is not an uncommon issue in epidemiology models; for example, it is standard practice to use a single-phase transition model component for processes such as incubation and recovery in spite of the fact that doing so tacitly assumes that times for such processes are exponentially distributed.  

In this study, we examine the extent of differences in model outcomes caused by adopting the naive assumption that all susceptibles are in the vaccination queue (with or without deprecating the rate constant to account for unwillingness) rather than the more realistic assumption that only some susceptibles are in that queue.  For a conservative approach, we allow for the possibility that benefits from vaccination are relatively minor.  Compared to susceptible individuals, we assume three benefits of vaccination: lower rate of infection, lower infectiousness, and faster recovery.

The significance of the decision on how to incorporate vaccination into the model may depend on the time frame for which the model is intended; hence, we separately consider scenarios for the endemic and epidemic phases of the disease.  In the latter case, the amount of time over which the model is tracked could affect the perceived importance of the vaccination modeling, so we examine how the effect of the difference progresses over time.  The key questions are these:
\benum
\item
How much error results from completely neglecting vaccine unwillingness?
\item
To what extent can the error be corrected by incorporating unwillingness into the vaccination rate constant rather than into the rate diagram?
\eenum
We examine how the answer depends on the basic reproduction number of the disease ($\rr$), the size of the unwilling fraction ($W$), the vaccination rate constant scaled with respect to disease duration ($\theta$), the relative susceptibility of vaccinated individuals ($\sigma<1$), the relative infectiousness of previously-vaccinated infectious individuals ($\eta<1$), and the relative recovery rate of previously-vaccinated individuals ($\alpha>1$). 

\section{An SIR Model with Vaccination}

Figure~\ref{fig381} shows an SIR model that includes vaccination and demographics.  The total birth rate is taken to match the death rate in order to keep the population at a constant value, which we can take without loss of generality to be 1.  In addition to classifying population members as susceptible, infectious, or removed, we classify individuals as either willing or unwilling/unable to accept vaccination, with fixed fractions $W$ and $U$.\footnote{See Jiang et al \cite{jiang2025} for a more nuanced model in which attitudes change in response to disease levels.}  \red{The model allows for three possible vaccination benefits, relative to unvaccinated susceptibles:
\benum
\item
The infection rate is multiplied by a factor $\sigma \le 1$;
\item
Infectiousness is multiplied by a factor $\eta \le 1$;
\item
Recovery rate is increased by a factor $\alpha \ge 1$.
\eenum
}
These three parameters combine into a single parameter that represents the relative disease burden of vaccinated individuals:
\beq
\label{xidef}
\xi = \frac{\sigma \eta}{\alpha}.
\eeq
It is reasonable to assume that this imperfect vaccination has a roughly equal quantitative impact for the three modifications.  Hence, when needed, we make the additional simplifying assumption that the effects of vaccination on the infection rate, infectiousness, and recovery time are equal, that is,
\beq
\alpha \sigma = 1, \qquad \eta=\sigma, \qquad \xi = \sigma^3.
\eeq
\red{To keep the models simple, we assume that disease recovery confers lifetime immunity.  Relaxing this assumption would increase the importance of vaccine unwillingness because recovered individuals whose immunity lapses would again be susceptible, when additional vaccination would again require willingness.}

To accommodate the differences between vaccinated and unvaccinated susceptibles, it is necessary to have three susceptible classes, one willing but pre-vaccinated ($S_W$), one vaccinated ($V$), and one unwilling ($S_U$) and to distinguish between infectious individuals who were ($I_V$) and were not ($I$) vaccinated.  For convenience, we define the total population infectiousness as
\beq
\ihat=I + \eta I_V
\eeq
and use this quantity to simplify the formulas.  The figure illustrates the more general endemic model with demographic processes of birth and natural death.  In this case, willingness is incorporated into the birth rates.

\begin{figure}[h]
\centering
\begin{tikzpicture}

	\draw[line width=1.5] (0,4) rectangle (1,5);
	\node at (0.5,4.5) {\LARGE{$S_U$}};
	
    \draw[line width=1.5][->](0.5,5.8)--(0.5,5);
	\node at (0.8,5.35) {\normalsize{$\mu U$}};

	\draw[line width=1.5][->](0.5,4)--(0.5,3.2);
	\node at (0.9,3.6) {\normalsize{$\mu S_U$}};

	\draw[line width=1.5] (0,7) rectangle (1,8);
	\node at (0.5,7.5) {\LARGE{$S_W$}};
	
    \draw[line width=1.5][->](0.5,8.8)--(0.5,8);
	\node at (0.8,8.35) {\normalsize{$\mu W$}};

	\draw[line width=1.5][->](0.5,7)--(0.5,6.2);
	\node at (0.9,6.6) {\normalsize{$\mu S_W$}};

    \draw[line width=1.5] (3,7) rectangle (4,8);
	\node at (3.5,7.5) {\LARGE{$V$}};

	\draw[line width=1.5][->] (1,7.5)--(3,7.5);
	\node at (2,7.75) {\normalsize{$\phi S_W$}};

    \draw[line width=1.5][->](3.5,7)--(3.5,6.2);
	\node at (3.8,6.65) {\normalsize{$\mu V$}};

    \draw[line width=1.5] (5,4) rectangle (6,5);
	\node at (5.5,4.5) {\LARGE{$I$}};
	 
	\draw[line width=1.5][->] (1,4.5)--(5,4.5);
	\node at (3,4.75) {\normalsize{$\beta S_U \ihat$}};
	 
	\draw[line width=1.5][->] (1,7.25)--(5,4.75);
	\node at (4.2,5.75) {\normalsize{$\beta S_W \ihat$}};

    \draw[line width=1.5][->](5.5,4)--(5.5,3.2);
	\node at (5.8,3.65) {\normalsize{$\mu I$}};

    \draw[line width=1.5] (6,7) rectangle (7,8);
	\node at (6.5,7.5) {\LARGE{$I_V$}};
	 
	\draw[line width=1.5][->] (4,7.5)--(6,7.5);
	\node at (5,7.75) {\normalsize{$\sigma \beta V \ihat$}};

    \draw[line width=1.5][->](6.5,7)--(6.5,6.2);
	\node at (6.9,6.65) {\normalsize{$\mu I_V$}};

    \draw[line width=1.5] (9,5.5) rectangle (10,6.5);
	\node at (9.5,6) {\LARGE{$R$}};
	
	\draw[line width=1.5][->](9.5,5.5)--(9.5,4.7);
	\node at (9.9,5.1) {\normalsize{$\mu R$}};

	\draw[line width=1.5][->](6,4.5)--(9,5.75);
	\node at (7.5,5.4) {\normalsize{$\gamma I$}};

	\draw[line width=1.5][->](7,7.5)--(9,6.25);
	\node at (8.2,7.2) {\normalsize{$\alpha \gamma I_V$}};

\end{tikzpicture}
\caption{An SVIR model with vaccination applied only to a willing subclass of susceptibles.
\label{fig381}}
\end{figure}
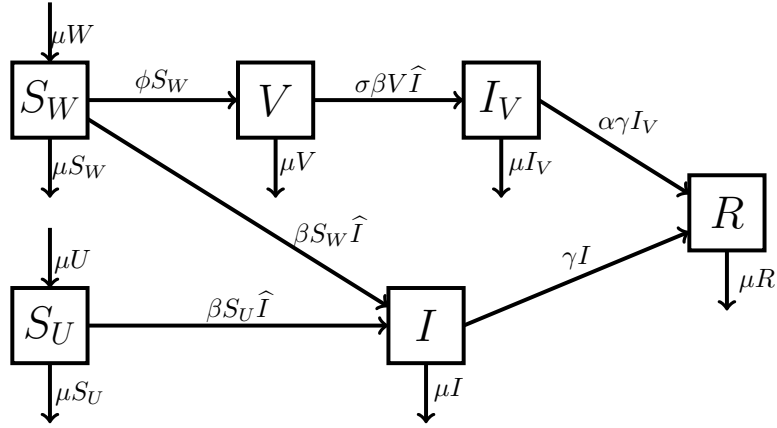

From the diagram, we have the model
\begin{align}
\begin{aligned}
\label{model1}
\dot{I} &= -(\gamma+\mu)I+\beta S \ihat, \\
\dot{I_V} &= -(\alpha \gamma+\mu)I_V+\sigma \beta V \ihat, \\
\dot{S} &= \mu-\mu S-\phi S_W-\beta S \ihat, \\
\dot{V} &= \phi S_W-\mu V-\sigma \beta V \ihat, \\
\dot{S_W} &= \mu W-(\phi+\mu) S_W-\beta S_W \ihat, \\
\end{aligned}
\end{align}
where we have used the full unvaccinated susceptible population $(S=S_W+S_U)$ as a variable in place of the unwilling subpopulation and use dots to indicate time derivatives.

\subsection{Scaling on an Epidemic Time Scale}

Full understanding of epidemiological models is facilitated by a proper scaling.  Because the population in our model is constant, we can assume without loss of generality that the model populations represent population fractions.  We still need to scale the time.  There are two possibilities: the mean lifespan $1/\mu$ and the mean infectious duration $1/(\gamma+\mu)$ \cite{ledderBMB}.  These scales yield slow and fast dimensionless time variables, given respectively as
\beq
\label{dimtime}
T = \mu t. \qquad \tau = (\gamma+\mu)t.
\eeq
For both scalings, we create dimensionless parameters using the fast time as a reference:
\beq
\epsilon = \frac{\mu}{\gamma+\mu} \ll 1, \qquad
\rr = \frac{\beta}{\gamma+\mu}, \qquad
\theta = \frac{\phi}{\gamma+\mu}, \qquad a=1-\alpha^{-1},
\eeq
with $a$ added merely for algebraic convenience.  For the epidemic scaling, we replace the time derivative using $d/dt=(\gamma+\mu) d/d\tau$ and then divide the resulting equations by the factor $\gamma+\mu$ to obtain the scaled model
\begin{align}
\begin{aligned}
\label{epidemic}
\frac{dI}{d\tau} &= -I+\rr S \ihat, & I(0)&=\epsilon I_0, \\
\frac{dI_V}{d\tau} &= -\alpha(1-\epsilon a)I_V+\sigma \rr V \ihat, & I_V(0)&=0, \\
\frac{dS}{d\tau} &= \epsilon(1-S)-\theta S_W-\rr S \ihat, & S(0)&=1-\epsilon I_0, \\
\frac{dV}{d\tau} &= \theta S_W-\epsilon V-\sigma \rr V \ihat, & V(0)&=0, \\
\frac{dS_W}{d\tau} &= \epsilon (W-S_w) -\theta S_W-\rr S_W \ihat, & S_W(0)&=W S(0), \\
\end{aligned}
\end{align}
where we have assumed an initial state in which nobody has been vaccinated or is removed from risk, so that the population consists only of a small number of infectives and the rest susceptibles.  The factor $\epsilon$ in the initial conditions codifies the smallness of the initial infectious population without loss of generality.

Without doing any analysis, the scaled model offers some insights into the relative importance of the different epidemiological and demographic processes.  In the epidemic scaling, all of the demographic terms (birth and natural death) are $O(\epsilon)$ relative to the epidemiological terms.  \red{Demographic processes are too slow to have much impact in the epidemic phase of a disease.}  Changes in the full susceptible class and the willing susceptible subclass are dominated by processes that decrease these populations.  We can consider the epidemic phase to have ended when these populations have been so decreased that the birth process is sufficient to start recovery of the susceptible population.  This suggests that the primary goal of epidemic analysis should be to identify the total population fraction that has escaped the epidemic at the time $\tau_\infty$, arbitrarily taken as the time when the total infection rate of susceptibles ($\rr S \ihat$) is equal to the birth rate ($\epsilon$); then we define the output of the epidemic model in terms of the willing population fraction as the function
\beq
\label{epifnc}
f(W;\rr,\sigma,\eta,\alpha,\theta,\epsilon,I_0)=S(\tau_\infty)+V(\tau_\infty).
\eeq
The function notation with the semicolon is used to identify $f$ as a family of functions of the single variable $W$, distinguished by a vector of parameters.  The function $f$ needs to be determined from numerical simulations.
An additional outcome of the epidemic analysis will be the graphs of $\ihat$ and $V$, showing the total infectiousness and vaccinated population size over time.  

For the modification of the model where we try to account for unwillingness without having two susceptibles classes, the $S$, $V$, and $S_W$ equations in \eqref{epidemic} are replaced by
\begin{align}
\begin{aligned}
\label{epidemic2}
\frac{dS}{d\tau} &= \epsilon(1-S)-W\theta S-\rr S \ihat, & S(0)&=1-\epsilon I_0, \\
\frac{dV}{d\tau} &= W\theta S-\epsilon V-\sigma \rr V \ihat, & V(0)&=0, \\
\end{aligned}
\end{align}
where we have used $W \theta$ as the vaccination rate constant so that the initial vaccination rate is $W \theta$, matching that of the more general model [$\theta S_W(0)=W \theta$], and the $S_W$ equation is replaced by $S_W=S$.

\subsection{Scaling on an Endemic Time Scale}

To scale the model for the endemic scenario, we apply the derivative transformation $d/dt=\mu \,d/dT$ to obtain the system
\begin{align}
\begin{aligned}
\epsilon \frac{dI}{dT} &= -I+\rr S \ihat, \\
\epsilon \frac{dI_V}{dT} &= -\alpha (1-\epsilon a) I_V+\sigma \rr V \ihat, \\
\frac{dS}{dT} &= 1- S-\epsilon^{-1} \theta S_W-\epsilon^{-1} \rr  S \ihat, \\
\frac{dV}{dT} &= \epsilon^{-1} \theta S_W- V-\epsilon^{-1}\sigma \rr V \ihat, \\
\frac{dS_W}{dT} &= W-(\epsilon^{-1}\theta+1) S_W-\epsilon^{-1} \rr S_W \ihat. \\
\end{aligned}
\end{align}

While the model is now dimensionless, it is not yet properly scaled because some of the variables are $O(\epsilon)$ rather than $O(1)$ \cite{ledderBMB}.  This is most easily seen in the $S$ equation, where it appears that $S$ is dominated by the processes that decrease its value, as is true in the epidemic phase.  In actuality, the variables $S_W$ and $\ihat$ are $O(\epsilon)$ on the endemic time scale, so that all three negative terms in the $S$ equation are comparable in importance.  To correctly scale all the state variables and simplify the notation, we make substitutions using the following definitions:
\beq
I = \epsilon Y, \qquad \eta I_V = \epsilon X, \qquad \rr \ihat = \epsilon Z, \qquad \theta S_W=\epsilon P.
\eeq
With these changes, and with $\sigma \eta$ replaced by $\alpha \xi$ using \eqref{xidef}, we obtain the scaled model
\begin{align}
\label{endemic}
\begin{aligned}
\epsilon Y' &= -Y+SZ, \\
\epsilon X' &= -\alpha X+\alpha \xi V Z+\epsilon \alpha a X, \\
S' &= 1-S-P-SZ, \\
V' &= P- V-\sigma V Z, \\
\epsilon P' &= \theta (W-P)-\epsilon P-\epsilon P Z, \\
\end{aligned}
\end{align}
with
\beq
\label{totalI}
Z=\rr(X+Y).
\eeq
There are no initial conditions because the purpose of the endemic scaling is to determine long-term behavior.

As in the epidemic version, there are insights to be gained from mere examination of the scaled model.  The primary processes governing changes in the willing susceptible population, represented by $P$, are birth and vaccination, with natural death and infection occurring at a slower rate and so having little impact.  On the endemic scale, any difference between $P$ and $W$ is resolved quickly.  This biological observation corresponds to the mathematical observation that the $P$ equation is decoupled from the rest (at leading order in the asymptotic regime $\epsilon \to 0$) when it comes to determining stability.

There are two primary outcomes of interest for the endemic model, depending on which of the two possible equilibria, disease-free (DFE) or endemic (EDE), is stable. If the benefit of vaccination when everyone is willing is sufficient to make the DFE stable, then the most important question to ask is how much willingness is required to stabilize the DFE; this minimum level of willingness will be seen to depend on the disease, represented by the basic reproduction number in the absence of vaccination ($\rr$), and the overall relative disease burden of the vaccinated, represented by $\xi$.  For cases where the EDE is stable, the basic reproduction number is not the primary endemic outcome of interest.  Instead we are interested primarily in the equilibrium rate of new infections in the asymptotic limit $\epsilon \to 0$; this is given by a function
\beq
F(W; \rr, \sigma, \eta, \alpha, \theta)=SZ+\sigma VZ = 1-S-V,
\eeq
which we'll see can be determined analytically.  Leading order approximation is quite accurate, as the parameter $\epsilon$ is very small.  For example, if the expected disease duration is about 2 weeks and the expected lifespan of the population is about 80 years, then $\epsilon \approx 0.0005$.  

\subsection{Parameter Estimates}

Given that the model is not for a specific disease, we want to consider a relatively wide range of parameter values.  \red{We consider basic reproduction values in the range $1 \le \rr \le 6$, corresponding to diseases of relatively slight to very significant infectiousness, although not including diseases as infectious as measles.  Most sources estimate the basic reproduction number for influenza at roughly 2--3, while the original strain of the COVID-19 virus had a basic reproduction number closer to 6 \cite{ke2021}, with later strains very likely higher.  The vaccination rate parameter is the ratio of mean disease duration to mean time for willing susceptibles to be vaccinated.  With a typical disease duration of 2 weeks and a range of 4 to 20 weeks for mean vaccination time, we have the range $0.1 \le \theta \le 0.5$.  The relative infection rate for vaccinated individuals ($\sigma$) can vary considerably.  It is 0 for a near-perfect vaccine, such as the measles vaccine, while considerably lower for the seasonal flu vaccine.  CDC data suggests that the worst case for the flu vaccine is a 30\% reduction, while a more typical case is 50\% reduction \cite{cdc2025fluvaccine}.  We refer to the values $\sigma=0.7$ and $\sigma=0.5$, corresponding to these percentages, as ``marginal'' and ``moderate'' vaccine effectiveness.  For simplicity, we assume that the impact of vaccination on the infection duration ($1/\alpha$) and transmission rate of infectives ($\eta$) is the same as the relative rate of infection ($\sigma$).}  As noted earlier, a reasonable estimate for $\epsilon$ is 0.0005.  The value of this parameter plays little role in the outcomes, as its effect is primarily on the amount of time between successive jumps in the infection rate after the first epidemic.

\section{Endemic Analysis}

Both the disease-free and endemic equilibria require the Jacobian matrix for stability determination; the general form of this matrix is   
\beq
\label{J}
J = \left[ \begin{array}{ccccc}
-\Gamma(1-\rr S) & \Gamma \rr S & \Gamma Z & 0 & 0 \\
\Gamma \alpha \xi \rr V  & -\Gamma \alpha (1-\xi \rr V) + \alpha a & 0 & \Gamma \alpha \xi Z & 0 \\
-\rr S & -\rr S & -(1+Z) & 0 & -1 \;\\
-\sigma \rr V & -\sigma \rr V & 0 & -(1+\sigma Z) & 1 \\
-\rr P & -\rr P & 0 & 0 & -\Gamma \theta -(1+Z) \\
\end{array} \right],
\eeq
where for notational convenience, we have defined the additional quantity $\Gamma = \epsilon^{-1}$.


\subsection{The Disease-Free Equilibrium (DFE) and Basic Reproduction Number}

By inspection, we see that the equilibrium equations with all disease states set to 0 yields the DFE
\beq
\label{DFE}
X=Y=0, \qquad V=P=\frac{W}{1-\epsilon \theta^{-1}}, \qquad S=1-P.
\eeq
The Jacobian for this solution is
\beq
\label{J}
J_{DFE} = \left[ \begin{array}{cc|rrc}
-\Gamma(1-\rr S) & \Gamma \rr S & 0 & 0 & 0 \\
\Gamma \alpha \xi \rr V  & -\Gamma \alpha (1-\xi \rr V)+\alpha a & 0 & 0 & 0 \\
\hline
-\rr S & -\rr S & -1 & 0 & -1\; \\
-\sigma \rr V & -\sigma \rr V & 0 & -1 & 1 \\
-\rr P & -\rr P & 0 & 0 & -\Gamma \theta -1\, \\
\end{array} \right],
\eeq
The lower right block of this Jacobian is triangular; therefore, its diagonals contain three of the eigenvalues, all negative.  Stability is determined by the upper left block.

The upper left block also contains the information needed to determine the basic reproduction number using the next generation method \cite{nextgen}, as it constitutes $\hat{F}-\hat{V}$, where
\beq
\hat{F}=\left[ \bray{cc}
\Gamma \rr S & \Gamma \rr S \\
\Gamma \alpha \xi \rr V & \Gamma \alpha \xi \rr V
\eray \right], \qquad
\hat{V}=\left[ \bray{cc}
\Gamma & 0 \\
0 & \alpha (\Gamma-a)
\eray \right].
\eeq
These combine to make the next generation matrix
\beq
\hat{F}\hat{V}^{-1}=\left[ \bray{cc}
\rr S & \alpha^{-1} (1-\epsilon a)^{-1} \rr S \\
\alpha \xi \rr V & (1-\epsilon a)^{-1} \xi \rr V
\eray \right].
\eeq
The largest eigenvalue of this matrix is the basic reproduction number:
\beq
\rv = \rr \left[ S+\frac{\xi V}{1-\epsilon a} \right]
= \rr \left[ 1-\frac{W}{1-\epsilon \theta^{-1}} +\frac{\xi W}{(1-\epsilon \theta^{-1})(1-\epsilon a)} \right],
\eeq
using \eqref{DFE} and the fact that the determinant of the matrix is 0, which makes the eigenvalues 0 and the trace of the matrix.
In the limit as $\epsilon \to 0$, this formula reduces to the asymptotic approximation
\beq
\label{Rv}
\rv \sim \rr [1-(1-\xi)W]= \rr (U+\xi W), \qquad U \equiv 1-W.
\eeq
This result was to be expected, as $\xi$ is the relative disease burden through universal vaccination and $W$ is the fraction of individuals who are willing to be vaccinated.  The small fraction of willing susceptibles who do not make it into class $V$ either because of prior infection or unrelated death accounts for the minor deviation of the exact result from the leading order approximation.  As a practical matter, deviations from leading order can be ignored, as they are far less significant than the smallest errors in estimation of parameter values.

Stability of the DFE can be checked using the Routh-Hurwitz conditions for the submatrix in the upper left corner, which we call $J_{12}$, identifying this as the submatrix constructed from rows 1 and 2 and columns 1 and 2.  The relevant quantities are the negative trace and the determinant, both of which must be positive; to leading order these are
\beq
c_1 = -\tr(J_{12})=\Gamma [1-\rr S+\alpha -\alpha \xi \rr V],
\eeq
\beq
c_2 = \det(J_{12})=\Gamma^2 \alpha [(1-\rr S)-\xi \rr V] = \Gamma^2 \alpha (1-\rr U-\xi \rr W)= \Gamma^2 \alpha (1-\rv).
\eeq
As expected, the condition $\rv < 1$ is necessary for DFE stability.  It remains to show that it is also sufficient.  To that end, we rearrange the inequality $c_2>0$ to obtain
\[ -\xi \rr V > \rr S-1 \]
and substitute this into the formula for $c_1$ to obtain
\begin{align*}
\epsilon c_1 &> 1-\rr S+\alpha +\alpha (\rr S-1) \\
&= 1+(\alpha -1) \rr S>0.
\end{align*}
Retaining the $O(\epsilon)$ terms makes for a minor change in the calculation without changing the final result that $\rv<1$ is the binding constraint for stability.

\subsection{The Endemic Disease Equilibrium (EDE)}

In epidemiology models that lack feedback, as is the case here, we should expect there to be a unique and asymptotically stable EDE that exists only if $\rv>1$.  Here we verify these expectations to leading order.

As a first step to finding the EDE, we use the $Y'$ and $X'$ equations from \eqref{endemic} to obtain equilibrium results
\[ Y=SZ, \qquad X=\xi VZ+\epsilon a X = (1+\epsilon a)\xi VZ + O(\epsilon^2). \]
The need to retain $O(\epsilon)$ terms at this stage will be seen in the stability analysis, where we need the key result obtained by substituting these equilibrium results into \eqref{totalI}:
\beq
\label{SplusV}
\rr S + \xi \rr V \sim 1-\epsilon a \xi \rr V.
\eeq
From here, leading order approximations in the calculation of the EDE will be sufficient.  

Using the leading order approximation $P \sim W$, the $S'$ and $V'$ equilibrium equations become
\beq
\label{SandV}
S(1+Z) = U, \qquad V(1+\sigma Z) = W.
\eeq
Combining the second of these equations with \eqref{SplusV} (to leading order) to eliminate $V$ yields
\[ (1-\rr S)(1+\sigma Z)=\xi \rr V (1+\sigma Z)=\xi \rr W = \rv-\rr U. \]
Rearranging this equation yields
\[ \rr (U-S) + \sigma (1-\rr S) Z = \rv-1. \]
Substituting from the first equation of \eqref{SandV} puts this result into the form
\beq
\label{Zeqn}
[\sigma+(1-\sigma) \rr S] Z = \rv-1.
\eeq
This result serves two purposes: It provides an explicit formula for $Z$ once $S$ is known, and it shows that $\rv>1$ is a necessary condition for the existence of an EDE.

We can now eliminate $Z$ from \eqref{Zeqn} by first multiplying by $S$ and then substituting again from \eqref{SandV}.  This yields a quadratic equation for $S$:
\beq
\label{Seqn}
(U-S)[\sigma+(1-\sigma) \rr S]=(\rv-1)S.
\eeq
Because $\rr S \le 1$, the quantity in the square brackets is also less than or equal to 1; hence,
\[ U-S \ge (\rv-1)S, \]
from which we obtain the bound $S \le \rv^{-1} U$.  Thus, there is a unique solution in the interval $0 \leq S \le \rv^{-1} U$ and so $\rv>1$ is sufficient as well as necessary for existence of a unique EDE.

Proof of the asymptotic stability of the EDE follows the outlines of a simplified form of the Routh-Hurwitz conditions, presented here.
\begin{theorem}
Let $J$ be the Jacobian matrix for a four-component dynamical system.  Let $c_1=-\tr(J)$, $c_2$ be the sum of all $2 \times 2$ subdeterminants of $J$, $c_3$ be the negative of the sum of all $3 \times 3$ subdeterminants of $J$, $c_4=\det{J}$, and combine these quantities as 
\[ q_1=c_1c_2-c_3, \qquad \Delta_1 = c_2-c_1, \qquad \Delta_2 = c_4-c_3, \qquad \rho=c_3/c_1. \]
\benum
\item
The associated equilibrium point is asymptotically stable if and only if
\beq
\label{RH1}
c_1>0, \qquad c_4>0, \qquad q_1>0, \qquad c_3q_1 > c_1^2c_4. 
\eeq
\item
If $c_1>0$, $c_3>0$, and $c_4>0$, then $q_2>0$ is necessary and sufficient for stability, where
\beq
\label{RH2}
q_2 = c_1c_3\Delta_1 -c_3^2-c_1^2 \Delta_2. 
\eeq
\item
The condition $q_2>0$ can be rewritten as
\beq
\label{RH}
\Delta_1>\rho+\frac{\Delta_2}{\rho}.
\eeq
\eenum
\end{theorem}

\begin{proof}
Claim 1 is the standard form of the Routh-Hurwitz stability criteria in four variables \cite{ledderasymptotics}.  The first part of claim 2 is easily seen from the fourth RH criterion, as $q_1 \le 0$ with $c_3q_1>c_1^2c_4$ is impossible when $c_3$ and $c_4$ are both positive.  Defining $q_2=c_3q_1-c_1^2c_4$, we have
\[ q_2 =c_1c_2c_3-c_3^2-c_1^2 (c_3+\Delta_2)=c_1c_3\Delta_1 -c_3^2-c_1^2 \Delta_2, \]
as required.  Replacing $c_3$ with $\rho c_1$ reduces the criterion $q_2>0$ to \eqref{RH}.
\end{proof}

As noted earlier, terms beyond the leading order as $\epsilon \to 0$ can be discounted, provided those leading order terms are not 0.  Thus, apparently lower order terms can be safely ignored only after checking that the apparently larger terms are non-zero.  Using the notation $J_{ij}$ for the submatrix of $J$ that uses rows $i$ and $j$ and columns $i$ and $j$, it appears from the structure of the Jacobian that $|J_{12}|$ is $O(\Gamma^2)$, which dominates the other terms that comprise $c_2$.  This is not correct, however, because of cancelation of dominant terms.  Hence, we must calculate $|J_{12}|$ carefully.  

We begin by introducing some simplifying notation into the Jacobian with
\beq
s=\rr S, \qquad v=\rr V, \qquad p=\rr P,
\eeq
yielding the result
\beq
\label{JEDE1}
J_{EDE} = \left[ \begin{array}{ccccc}
-\Gamma(1-s) & \Gamma s & \Gamma Z & 0 & 0 \\
\Gamma \alpha \xi v  & -\Gamma \alpha (1-\xi v-\epsilon a) & 0 & \Gamma \alpha \xi Z & 0 \\
-s & -s & -(1+Z) & 0 & -1 \;\\
-\sigma v & -\sigma v & 0 & -(1+\sigma Z) & 1 \\
-p & -p & 0 & 0 & -\Gamma \theta -(1+Z) \\
\end{array} \right],
\eeq

Using \eqref{SandV}, we can simplify the $j_{11}$ and $j_{22}$ entries of this matrix to
\[ j_{11}=-\Gamma(1-s)=-\Gamma (1+\epsilon a) \xi v+O(\epsilon), \]
\[ j_{22}=-\Gamma \alpha (1-\xi v-\epsilon a)=-\Gamma \alpha [s-\epsilon a (1-\xi v)] +O(\epsilon)=-\Gamma \alpha (1-\epsilon a) s+O(\epsilon). \]
From these formulas, we have
\[ |J_{12}|=\Gamma^2 \alpha \xi vs (1-\epsilon^2 a^2)-\Gamma^2 \alpha \xi vs=O(1). \]
Thus, the term $|J_{12}|$, which appears at first to dominate the other $2 \times 2$ subdeterminants, actually makes no contribution at the leading order level of $O(\Gamma)$.  With this detail in hand, we can now safely remove all terms that are subdominant in their own row, yielding the simplified matrix
\beq
\label{JEDE}
J_{EDE} = \left[ \begin{array}{cccc|c}
-\Gamma \xi v \;\; & \Gamma s & \Gamma Z & 0 & 0 \\
\Gamma \alpha \xi v  & -\Gamma \alpha s\;\; & 0 & \Gamma \alpha \xi Z & 0 \\
-s\;\; & -s\;\; & -(1+Z)\;\; & 0 & -1\;\; \\
-\sigma v\;\; & -\sigma v\;\; & 0 & -(1+\sigma Z)\;\; & 1 \\
\hline
0 & 0 & 0 & 0 & -\Gamma \theta\;\; \\
\end{array} \right],
\eeq
We see from the block structure that the matrix decomposes into a $4 \times 4$ and a $1 \times 1$ block, so we need only consider the larger block, to which Theorem 1 applies.  The quantity $c_1$ is the negative of the trace, which to leading order is
\[ c_1=\Gamma (\alpha s+\xi v). \]
Among the $3 \times 3$ subdeterminants, only $|J_{123}|$ and $|J_{124}|$ are of $O(\Gamma^2)$.  Calculating these subdeterminants yields the results
\[ J_{123} = -\Gamma^2 \alpha s Z (s+\xi v) \sim -\Gamma^2 \alpha s Z, \quad J_{124} = -\Gamma^2 \alpha \sigma \xi v Z (s+\xi v) \sim -\Gamma^2 \alpha \sigma \xi v Z. \]
Thus,
\beq
c_3=-|J_{123}|-|J_{124}|=\Gamma^2 Z (\alpha s +\alpha \sigma \xi v).
\eeq
Similarly, the determinant calculation for $J_{1234}$ yields the result
\beq
c_4=\Gamma^2 Z(\alpha s+\alpha \sigma \xi v+\alpha \sigma Z)=c_3+\Gamma^2 \alpha \sigma Z^2.
\eeq
Using the relation $\alpha \sigma = 1$, these last two formulas simplify to
\beq 
c_3 = \Gamma^2 (\alpha s +\xi v) Z = \Gamma Z c_1, \qquad c_4=c_3+\Gamma^2 Z^2;
\eeq
hence,
\beq
\rho = \Gamma Z, \qquad \Delta_2=\Gamma^2 Z^2=\rho^2.
\eeq
The  $c_4=c_3+\Gamma^2 Z^2$, so $\Delta_2=\Gamma^2 Z^2$.  Using the results for $\rho$ and $\Delta_2$, the stability criterion \eqref{RH} reduces to 
\beq
c_2>c_1+2\Gamma Z.
\eeq
Now examining the $2 \times 2$ subdeterminants, and using $\alpha \sigma=1$, yields the results
\[ |J_{13}|+|J_{24}|=c_1+2\Gamma Z, \qquad |J_{23}|+|J_{14}|>0, \qquad |J_{12}|=o(\Gamma), \qquad |J_{34}|=o(\Gamma), \]
which proves that the stability criterion is satisfied.  As is typical for straightforward epidemiology models, the EDE is stable whenever it exists, which is when $\rv>1$.

To summarize, we have established the needed analytical results for the EDE:
\bthm
There exists a unique and asymptotically stable EDE for the problem \eqref{endemic} if and only if $\rv>1$.  This solution is determined to leading order as $\epsilon \to 0$ by
\[ P = W, \quad Z = \frac{\rv -1}{\sigma+(1-\sigma)\rr S}, \quad V = \frac{W}{1+\sigma Z}, \quad Y = SZ, \quad X = \xi VZ, \]
with $S$ the unique solution of
\[ (U-S)[\sigma+(1-\sigma)\rr S] =(\rv-1)S, \qquad \rv = \rr (U+\xi W) \]
in the interval $0 \le S \le \rv^{-1} U$.
The equilibrium rate of new infections is then
\[ F(W;\rr,\xi,\sigma)=1-S-V. \]
\ethm

It is interesting to note that the vaccination rate constant $\theta$ does not appear in either the formula for $\rv$ or the EDE formulas; hence, adjusting this parameter in an attempt to account for vaccine unwillingness would make no difference.

\section{Results}

\subsection{The Endemic Scenario}

\begin{figure}[ht]
\centering    
\includegraphics[width=0.75\textwidth]{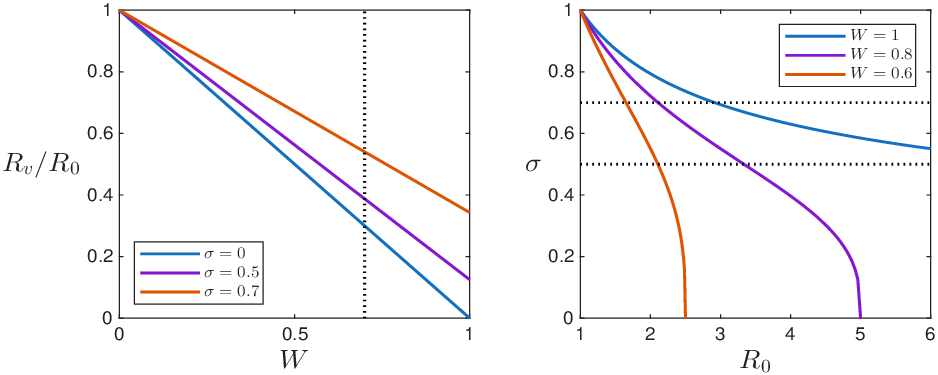}
\caption{Left: Reduction factor from $\rr$ to $\rv$; Right: Regions in disease-vaccine space where the DFE is stable (below and left of the curves)}
\label{SVIR_Rv}
\end{figure}

On the endemic time scale, our tasks were to determine the impact of vaccine unwillingness on the basic reproduction number \eqref{Rv} and the equilibrium new infection rate $F(W)$ given in Theorem 2.  It is reasonable to reduce the number of parameters by assuming that all vaccination modifications are of equal value, so that $\xi = \sigma^3$.  Figure \ref{SVIR_Rv} displays the results.  \red{The left panel of the figure shows the factor by which the basic reproduction number is reduced by vaccination as a function of the fraction of the population that is willing to be vaccinated.  The three curves are for values of $\sigma$ corresponding to a perfect vaccine, a moderate vaccine ($\sigma=0.5$), and a marginal vaccine ($\sigma=0.7$).  At first it seems odd that the moderate curve is closer to the perfect one than the marginal one; the reason for this is that the basic reproduction number is a linear function of the overall epidemiological disease burden of vaccinated individuals, which is $\xi=\sigma^3$, rather than a linear function of the infection reduction factor $\sigma$.  The vertical dotted line marks $W=0.7$, which is a reasonable estimate of vaccine willingness.  Even with the minimal vaccine, the difference in outcome between $W=1$ and $W=0.7$ is large, with an $\rv:\rr$ ratio of 0.34 for the former and 0.54 for the latter.}

The basic reproduction number itself is not the most important property of an endemic disease model; rather, it is the hypersurface in the parameter space that marks out the stability regions for the DFE and EDE.  Several projections of this hypersurface onto the $\rr \sigma$ plane are shown in the right panel of Figure \ref{SVIR_Rv}, with three different values of willingness.  The region below and to the left of each curve indicates the disease/vaccine combinations where the DFE is stable, given the level of willingness indicated in the legend.  Thus, the region between the $W=1$ and $W=0.8$ curves shows disease/vaccine combinations for which the assumption of full willingness gives the wrong stability result if the actual willingness is only 80\%.  The dotted lines mark the approximate range of $\sigma$ values for the seasonal flu vaccine \cite{cdc2025fluvaccine}.  Most other vaccines are far more effective, corresponding to lower $\sigma$.  Even 10\% unwillingness would significantly overestimate the impact of vaccination on disease eradication, particularly for a highly effective vaccine.

\begin{figure}[ht]
\centering    
\includegraphics[width=0.35\textwidth]{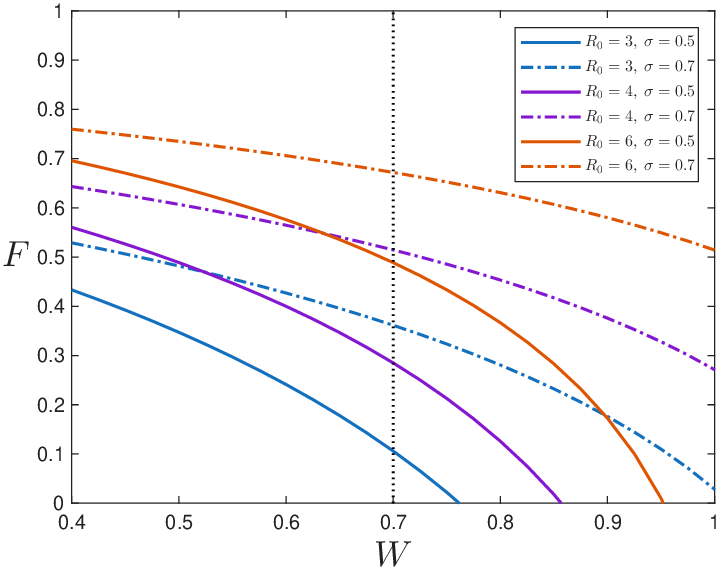}
\caption{The rate of new infections at equilibrium}
\label{infrate}
\end{figure}

When the EDE is stable, the most important model outcome is the rate of new infections at equilibrium.  This outcome is illustrated in Figure \ref{infrate}.  \red{The impact of assuming $W=1$ when the actual willingness is $W=0.7$ can be seen by comparing the values of $F$ on the dotted vertical line with those at the far right of the plot.  This impact is significant for a vaccine that is only marginally effective (the dash-dot curves), and grows quickly as the vaccine effectiveness increases.}  As in the case when the DFE is stable, neglecting vaccine unwillingness yields results that are too far from correct to be useful.

Instead of making the rate diagram more complicated to address vaccine unwillingness, couldn't we instead compensate for unwillingness by replacing the rate constant $\phi$ with the smaller constant $W \phi$ while retaining the simpler disease structure that assumes full willingness?  For endemic scenarios, the answer is an emphatic ``no.''  Changing the rate constant makes no difference to leading order in the endemic model, which is essentially the same as making no difference at all.  While our study considers a particularly simple disease profile, there is no reason to think that this conclusion would be any different for a more complicated profile.  Given that unwillingness is always a significant issue in any community, all endemic models with vaccination need to have unwillingness built into the rate diagram.

\subsection{The Epidemic Scenario}

It was surprising that the rate constant $\theta$ played no role in the endemic results to leading order.  It should play a significant role at leading order in the epidemic model, however.  The long-term outcome of the epidemic model is defined as \eqref{epifnc}, while the short term behavior can be defined as the graphs of $\ihat$ and $V$.  In both cases, our primary goal is to observe the difference made by attempting to account for unwillingness by deprecating the rate constant from $\theta$ to $W\theta$ rather than by adding an unwilling susceptible subclass.  This comparison depends on the disease and vaccine characteristics $\rr$ and $\sigma$.  

Note that it would be a simple matter to compute sensitivity coefficients for various parameters to various outcomes; however, this is overkill given that the model is for a generic disease with parameters estimated from general averages.  Instead, it is more useful to judge impact as ``noticeable'' if one can see a difference in a graph and ``significant'' if the curves show a clear difference over a broad range of parameter sets.

\begin{figure}[ht]
\centering    
\includegraphics[width=0.75\textwidth]{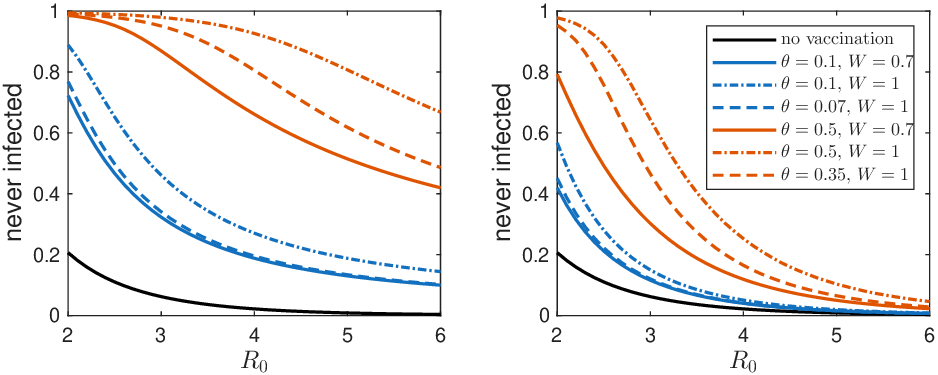}
\caption{Population fractions that escape infection in the epidemic scenario, with $\sigma=0$ (left) and $\sigma=0.7$ (right).  The nominal vaccination rate constants are $\theta=0.1$ (blue) and $\theta=0.5$ (orange).  Solid curves have $W=0.7$, and dash-dot curves have $W=1$.  Dashed curves are for the model that assumes full willingness but deprecates the vaccination rate constant by a factor of $W=0.7$.}
\label{safe}
\end{figure}

Figure \ref{safe} shows how the graph of the population fraction that avoids infection during the initial epidemic \eqref{epifnc} vs $\rr$ depends on a variety of factors.  The left panel is for the perfect vaccine ($\sigma=0$) and the right for a marginal vaccine ($\sigma=0.7$).  The bottom (black) curve is for the case of no vaccine, the middle group (blue) is for slow vaccination ($\theta=0.1$), and the top group (orange) is for fast vaccination  ($\theta=0.5$).  Within each group, the three curves represent three different assumptions about willingness.  The lowest (solid) curve is for the full model \eqref{epidemic} with willingness used to partition the susceptible class into subgroups.  The top (dash-dot) curves are for the model without subclasses \eqref{epidemic2} and assuming full willingness.  The middle (dashed) curves are the more nuanced simplification that uses the model without subclasses but attempts to correct for unwillingness by using $W \theta$ as the infection rate constant for the model.

\red{Comparison of the solid and dash-dot curves shows the impact of vaccine unwillingness.  Significantly more people become infected during the epidemic when 30\% of people are unwilling to be vaccinated, with the effect larger for more effective vaccines and/or faster vaccination rate.  These details are not surprising, as vaccination should have a greater impact when it is accomplished quickly and when it is more effective.}

\red{The overall impact of vaccine effectiveness is seen in the general comparison of the two panels in the figure.  A perfect vaccine has a very large impact even if there is significant unwillingness and the vaccination rate is slow.  However, the impact of a marginal vaccine, delivered at a slow rate, and with 30\% unwillingness (the solid blue curve in the right panel) is still noticeable for diseases of moderate infectiousness ($\rr<4$) and significant for any disease if the vaccination rate is large, even with 30\% unwillingness.  Of the factors that determine the public health consequences of vaccination, it appears from the figure that disease infectiousness and vaccine effectiveness are the most important, followed by vaccination rate and then unwillingness.}

What then of the modeling strategy of accounting for unwillingness by deprecating the rate constant rather than using the more correct compartment diagram (solid curves vs dashed curves)?  This depends on the vaccination rate.  For slow vaccination (blue curves), the difference is noticeable, but not large.  For fast vaccination (orange curves), the difference is always significant, regardless of vaccination effectiveness or disease infectiousness.

\begin{figure}[ht]
\centering    
\includegraphics[width=0.75\textwidth]{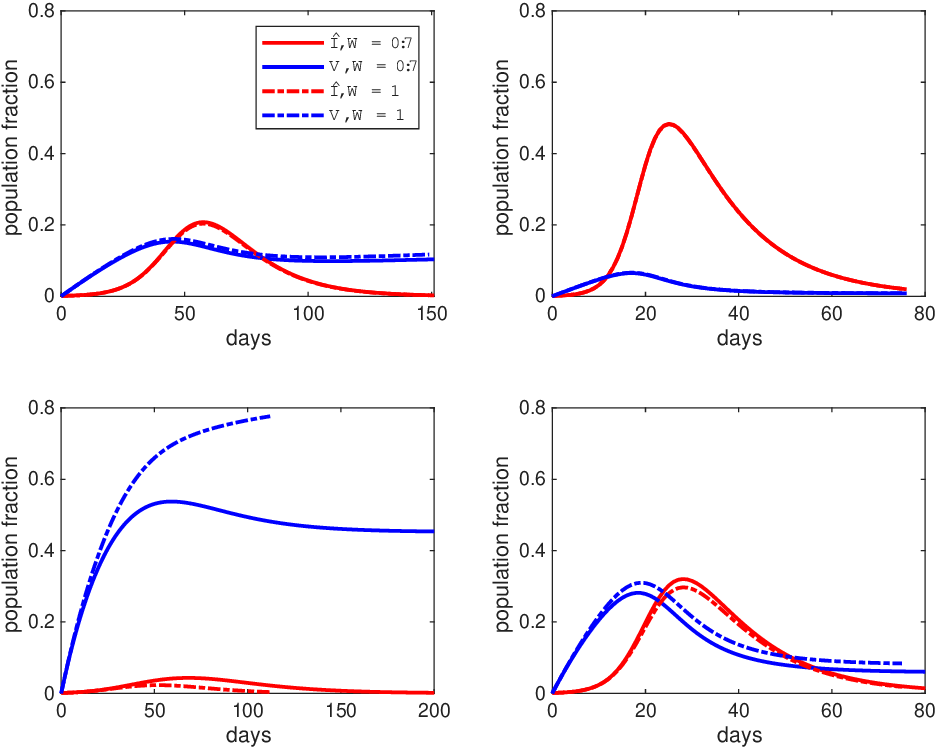}
\caption{Total population infectiousness $\ihat$ (red) and vaccinated fraction $V$ (blue), with solid curves for the model \eqref{epidemic} and dash-dot curves for the simplified model that accounts for unwillingness by decreasing the vaccination rate constant to $W \theta$ rather than modifying the standard SVIR compartment model; $\rr=3$ (left panels), $\rr=6$ (right panels), $\theta=0.1$ (top panels), $\theta=0.5$ (bottom panels); all panels have $\sigma=0.5$ and $W=0.7$.}
\label{SVIR_expts}
\end{figure}

For a look at some specific simulation details, we examine graphs of the total infectiousness $\ihat$ and vaccinated population fraction $V$ for several scenarios, shown in Figure \ref{SVIR_expts}.  The solid curves are for \eqref{epidemic}, while the dash-dot curves are for the simplified model with the $S$, $S_W$, and $V$ equations replaced by \eqref{epidemic2}.  All scenarios use modest values of $\sigma=0.5$ and $W=0.7$.  The left panels have $\rr=3$, while the right panels have $\rr=6$; the top panels have $\theta=0.1$ and the bottom have $\theta=0.5$.

\red{Comparing the left and right panels shows the impact of the basic reproduction number.  As expected, larger $\rr$ results in a larger wave of infection.  Larger $\rr$ also decreases the size of the vaccinated population because a larger number of susceptibles become infected before they can be vaccinated.  Comparing the top and bottom panels shows the impact of the vaccination rate constant.  As expected, faster vaccination reduces the wave of infection in addition to increasing the vaccinated population.}  Of particular interest is the comparison between the two methods of accounting for unwillingness.  When the vaccination rate is slow, there is little difference between accounting for unwillingness by modifying the rate diagram and accounting for it by lowering the vaccination rate constant.  When vaccination is fast, there is error in both $\ihat$ and $V$ that is noticeable but not large, with the exception that the error in the size of the vaccinated fraction is large if $\rr$ is small in addition to the vaccination rate being fast.

\section{Discussion}

This study presents a simple model for a generic disease with an available vaccine and no mortality.  We set out to address the questions of how much error in key epidemiological outcomes is caused by ignoring vaccine unwillingness and the extent to which that error can be corrected by using vaccine unwillingess to deprecate the vaccination rate constant rather than to change the rate diagram for the disease.  The results are quite different, depending on the time scale of interest.

\red{For endemic models, whose purpose is to determine whether the final outcome is a stable disease-free equilibrium state or a stable endemic disease equilibrium state, the error caused by neglecting unwillingness is considerable and cannot be corrected by accounting for unwillingness in the vaccination rate constant.  On the endemic time scale of demographics, vaccination occurs so rapidly that its exact rate does not affect the key outcomes.  Nearly all susceptibles in the long-term scenario are unvaccinated, as willing susceptibles become vaccinated much more quickly than susceptible individuals of either willingness become infected.}

The answers to the questions about the importance of correct modeling of vaccination are more subtle in the epidemic time scale, where the large population of susceptibles means that a significant number of willing susceptibles become infected before they can be vaccinated.  \red{The epidemiological outcomes, and also the extent to which incorrect modeling matters, depend on the comparison between the rates by which susceptibles become infected, represented by $\rr$, and by which willing susceptibles become vaccinated, represented by $\theta$.  If a disease is highly infectious and the vaccination rate relatively slow, then vaccination has only a limited impact on the epidemiological outcomes, and consequently the error of ignoring unwillingness is unimportant.  In the opposite case of less infectiousness and faster vaccination, the error in neglecting unwillingness becomes significant.}

From a modeling perspective, it is clear that correct incorporation of vaccine unwillingness into models intended to determine the effective basic reproduction number, stability of equilibrium outcomes, and ongoing disease burden for endemic disease settings is absolutely vital.  For models intended to determine short-term outcomes, such as the population fraction that escapes infection in the initial wave, there is often no great harm in factoring unwillingness into the rate constant.  However, the analysis of models on the epidemic time scale is necessarily done by numerical integration of the differential equations.  While adding additional classes to an endemic model increases the complexity of the analysis, it makes no real difference to numerical integration whether we add a few additional variables.  Hence, it should be preferable to correctly account for vaccine unwillingness for epidemic models as well as endemic models.

\bibliographystyle{plain}
\bibliography{bibliography}

\end{document}